\documentclass{article}
\usepackage[utf8]{inputenc}
\usepackage{tikz-qtree,amsmath,amssymb,amsthm}

\title{Traversal-invariant characterizations of logarithmic space}
\date{}
\author{Siddharth Bhaskar\footnote{DIKU, University of Copenhagen, Copenhagen, Denmark}, Steven Lindell\footnote{Department of Computer Science, Haverford College, Haverford, PA, USA}, and Scott Weinstein\footnote{Department of Philosophy, University of Pennsylvania, Philadelphia, PA, USA}}

\begin{document}

\maketitle

\newtheorem{definition}{Definition}

\newtheorem{theorem}{Theorem}
\newtheorem{lemma}{Lemma}
\newtheorem{corollary}{Corollary}

\begin{abstract}
    We give a novel descriptive-complexity theoretic characterization of L and NL computable queries over finite structures using \emph{traversal invariance}. We summarize this as (N)L = FO + (breadth-first) traversal-invariance.
\end{abstract}

\section{Presentation invariance}

A common phenomenon in mathematics is that some property or quantity is defined in terms of some additional structure, but ends up being invariant of it. Dimension of a vector space and Euler characteristic of a manifold are important examples of this phenomenon; they are defined in terms of a given basis or simplicial complex respectively, but are invariant of the particular one chosen.

This state of affairs is very common in descriptive complexity theory. For example, suppose we want to compute the parity of a given finite set $X$. If we are given some linear ordering $(X,<)$, there is an inductive program computing the parity of $X$, but the result computed is independent of the particular ordering. Therefore, we call parity \emph{order-invariant LFP}: computable by an LFP program with a given order, but independent of the specific choice.

The celebrated result of Immerman and Vardi \cite{Imm86,Var82} that LFP logic captures polynomial-time queries over families of ordered finite structures can be recast as, \emph{order-invariant} LFP logic captures polynomial-time queries over \emph{all} finite structures.
Since then, a wide array of correspondences have been identified between known complexity classes on one hand, and invariant forms of LFP, MSO, or first-order logic on the other. For example, first-order logic and LFP logic invariant in arbitrary numerical predicates captures $\mathrm{AC}^0$ and $\mathrm{P}/_\mathrm{poly}$ respectively  \cite{Sch13}.

\paragraph{Our contribution}
We give a novel characterization of logarithmic and nondeterministic logarithmic space queries using presentation-invariant first-order definability (Theorems \ref{main-1} and \ref{main-2}). The presentations in question are traversals and breadth-first traversals respectively, which are certain types of linear orders on finite graphs.

This is to our knowledge the first characterization of L or NL that does not rely on any sort of recursion or sequential computation, however limited, such as a function algebra, fixed-point logic, programming language, or automaton.

We find it fascinating and mysterious that passing from traversals to breadth-first traversals in the presentation causes a jump from L to NL in definability power. It begs the question, what other complexity classes can be characterized by certain types of graph search?

\paragraph{Structure of this paper}
In Section \ref{sec:traversals}, we discuss traversal- and breadth-first traversal-invariant definability, and show the definability of undirected and directed reachability respectively. In Section \ref{sec:descriptive complexity}, we present descriptive-theoretic characterizations of L and NL (Theorems \ref{main-1} and \ref{main-2}).

\paragraph{Preliminaries and notation}
We assume familiarity with basic graph theory, automata theory, and model theory, including the notion of \emph{interpretation}. We will denote graphs and other first-order structures by uppercase Roman letters. By ``graph'' we always mean ``undirected graph;'' we will say ''directed graph'' when we need to.
We denote families of structures in a common signature by captial calligraphic letters, e.g., $\mathcal{K}$.

\if{false}

\subsection{Definitions}

Our basic objects of study are \emph{queries over families of structures}. Let $\mathcal{L}$ be a family of structures of a common signature $L$. Then an $n$-ary query $R$ is a relation $R^A \subseteq A^n$ for each structure $A \in \mathcal{L}$ that satisfies the isomorphism-invariance condition, for all $A, B \in \mathcal{L}$ and isomorphism $f$ mapping $A$ onto $B$, $f[R^A] = R^B$; we write $A \models R(\vec{x})$ to mean $\vec{x} \in R^A$. In the important special case that $n=0$, $R$ can be interpreted as a subset of $\mathcal{L}$; this is called a \emph{boolean query}.

Invariant queries are always parameterized by pairs of families of structures $\mathcal{(L,P)}$. Whenever we form such a pair, we implicitly assume:

\begin{enumerate}
    \item The signature $L$ of $\mathcal{L}$ is contained in the signature $L^+$ of $\mathcal{P}$,
    \item $\mathcal{L}$ and $\mathcal{P}$ are nonempty,
    \item every structure in $\mathcal{L}$ has some expansion in $\mathcal{P}$, and 
    \item the $L$-reduct of every structure in $\mathcal{P}$ is in $\mathcal{L}$.
\end{enumerate}

\begin{definition}
A query $R$ over $\mathcal{P}$ is \emph{$(\mathcal{L},\mathcal{P})$-invariant} in case for any two structures $A$ and $B$ in $\mathcal{P}$ with the same domain, if $A| L = B|L $, then $A \models R(\vec{x}) \iff B \models R(\vec{x})$ for any tuple $\vec{x}$ in their common domain.
\end{definition}

We will use terms, e.g., ``first-order $\mathcal{(L,P)}$-invariant formula,'' to indicate a first-order $L^+$-formula which defines a $\mathcal{(L,P)}$-invariant query.

In previous work, e.g., \cite{Ros07} another definition of presentation invariance is given:

\begin{definition}
 \emph{(Old definition)} Let $L$ and $L'$ be two signatures, $\mathcal{C}$ a class of $L$-structures and $\mathcal{Q}$ a class of $L'$ structures. A $L\cup L'$-query $R$ is $\mathcal{(C,Q)}$-invariant in case for any two $L\cup L'$-structures $A,B$ with the same domain, if $A|L = B|L$ and $A|L',B|L' \in \mathcal{C}$, then $A\models R(\vec{x}) \iff B \models R(\vec{x})$ for any $\vec{x}$ in their common domain.
\end{definition}

In the old definition, the structures and their presentations are ``independent of each other.'' This does not accommodate presentations of graphs by traversals, a central topic of this paper. 

To see how we generalize the old definition, take $L^+$ to be $L \cup L'$, take $\mathcal{L}$ to be $\mathcal{C}$, and take $\mathcal{P}$ to be the set of all $L^+$-structures whose $L$- and $L'$-reducts are in $\mathcal{C}$ and $\mathcal{Q}$ respectively. For a formula to have a well-defined semantics, for every structure in $\mathcal{C}$, there needs to be a structure in $\mathcal{Q}$ of the same cardinality. Hence every structure in $\mathcal{L}$ has some expansion in $\mathcal{P}$.

\subsubsection{``Syntax'' and semantics}
Suppose that $R$ is a $\mathcal{(K,P)}$-invariant query over $\mathcal{P}$. Then $R$ naturally induces a query over $\mathcal{K}$ in the natural way: take any structure in $\mathcal{K}$, arbitrarily expand it to a structure in $\mathcal{P}$, and interpret $R$. By invariance, the answer is independent of the particular expansion.

When $R$ is defined by some sort of formula in the language $K^+$, it will be useful to syntactically distinguish the $\mathcal{P}$-query $R$ from the query it induces over $\mathcal{K}$.
To do this, we introduce a new type of quantifier called the \textbf{invariance quantifier}, really a family of quantifiers parameterized by the pair $\mathcal{(K,P)}$. If $\vec{f}$ is the list of nonlogical symbols in $K^+ \setminus K$, then we apply the invariance quantifier $\mathfrak{P}$ to the query $R$,  binding the symbols $\vec{f}$ to form the query $(\mathfrak{P}\vec{f})\,R$.

This has the semantics
$$ A \models (\mathfrak{P}\vec{f})\,R \iff \exists B \in \mathcal{P}\ (B|K = A) \wedge 
   (B \models R), $$
for any $A \in \mathcal{K}$. By $\mathcal{(K,P)}$-invariance of $R$, 
$$ A \models (\mathfrak{P}\vec{f})\,R \iff \forall B \in \mathcal{P}\ (B|K = A) \wedge 
   (B \models R). $$
Note that we are using the assumption that $\mathcal{P}$ is nonempty, otherwise these would not be equivalent.
   
While it is useful to have a mechanism distinguishing $R$ from $(\mathfrak{P}\vec{f})\,R$, this is only ``syntax'' in the loosest sense of the word, since it is highly non-effective to check that a given string is a well-formed formula.

It is easy to check, but useful to note, that the invariance quantifier satisfies the following properties, for any $\mathcal{(K,P)}$-invariant queries $R$ and $Q$:
$$ A \models (\mathfrak{P}\vec{f})\,(R \wedge Q) \iff A \models (\mathfrak{P}\vec{f})\, R \wedge A \models (\mathfrak{P}\vec{f})\, Q $$
$$ A \models (\mathfrak{P}\vec{f})\,(R \vee Q) \iff A \models (\mathfrak{P}\vec{f})\, R \vee A \models (\mathfrak{P}\vec{f})\, Q $$
$$ A \models (\mathfrak{P}\vec{f})\,\neg R \iff  A \not\models (\mathfrak{P}\vec{f})\, R $$
$$ A \models (\mathfrak{P}\vec{f})\,\exists x R \iff A \models \exists x\, (\mathfrak{P}\vec{f}) R$$
$$ A \models (\mathfrak{P}\vec{f})\,\forall x R \iff A \models \forall x\, (\mathfrak{P}\vec{f}) R$$

We have used $\mathfrak{P}$ to stand for a generic invariance quantifier here, but we shall reserve certain letters for specific presentations below.
   
\subsubsection{Basic properties}

\begin{lemma} Suppose $L\subseteq L^+$, $\mathcal{P}$ is a nonempty class of $L^+$-structures, and $\mathcal{L}$ its class of reducts to $L$. Then,
\begin{enumerate}
    \item Any first-order $L$-formula is $\mathcal{(K,P)}$-invariant.
    \item Conversely, if $\mathcal{P}$ is an elementary class, then any $(\mathcal{L},\mathcal{P})$-invariant first-order formula is equivalent over $\mathcal{L}$ to a first-order $L$-formula.
    \item If there is a uniform definitional expansion of structures in $\mathcal{L}$ to structures in $\mathcal{P}$, then any $(L,\mathcal{P})$-invariant first-order formula is equivalent to some first-order $L$-formula over $\mathcal{L}$. 
    \item If $\mathcal{K'} \subseteq \mathcal{L}$ and $\mathcal{P' \subseteq P}$, then any $\mathcal{(K,P)}$-invariant formula is also $\mathcal{(K',P')}$-invariant.
\end{enumerate}
\end{lemma}

The proofs of each of these are both easy and elementary, except for 2., which is immediate by Beth definability. By ``uniform definitional expansion'' in 3 we mean a $L$-formula for each atomic formula of $L^+$ such that if we take any structure in $\mathcal{L}$ and expand it by the relations defined by these formulas we get a structure in $\mathcal{P}$.

\subsection{Interpretations and change of signature}
A $(L,\mathcal{P})$-invariant formula defines a global relation over the family of $L$-structures. We can make sense of certain presentations, like linear orders, no matter what the base class $\mathcal{L}$ is; in that case, $\mathcal{P}$ would be the family of expansions of structures in $\mathcal{L}$ by all possible linear orders. However, the ``flagship presentation'' of this paper will be that of a \emph{traversal}, which only makes sense over graphs.

But we will want a notion of traversal-invariant definability over structures in an arbitrary signature. We do this by closing under elementary interpretations. (Indeed, even if we were only considering definability over graphs, we would want to close under elementary interpretations from graphs to graphs to recover our main descriptive complexity-theoretic results.)

We review the basic definitions behind interpretations in this section, roughly following the style of Hodges \cite{Hod93}. There is no new mathematical content here; however, getting the definitions and terminology straight is terribly important, since we will be using interpretations extensively. 

\begin{definition}
 Let $L$ and $K$ be signatures and $k \in \mathbb{N}$. An \emph{elementary $k$-ary interpretation} $\pi : L \to K$ is a first-order $K$-formula $\partial^\pi(\bar{x})$, for each constant symbol $c \in L$ a variable-free $K$-term $c^\pi$, and for each relation symbol $r \in L$, a first-order $K$-formula $r^\pi(\bar{x}_1,\dots,\bar{x}_n)$, where $n$ is the arity of $r$, and the length of each tuple throughout is $k$.
 
 An interpretation is \emph{quantifier-free} in case $\partial^\pi$ and each $r^\pi$ is quantifier-free.
\end{definition}

\begin{definition}
 Given an elementary $k$-ary interpretation $\pi : L \to K$ and a first-order $L$-term or $L$-formula $\alpha$, its \emph{translation} $\alpha^\pi$, is a $K$-formula given by the following recursion:
 \begin{enumerate}
    \item If $\alpha$ is a variable $x$, then $\alpha^\pi$ is a $k$-tuple of (distinct) variables $\bar{x}$.
    \item If $\alpha$ is a constant symbol $c$, then $\alpha^\pi$ is $c^\pi$. 
     \item If $\alpha$ is the atomic formula $r(\alpha_1,\dots,\alpha_n)$, then $\alpha^\pi$ is $r^\pi(\alpha^\pi,\dots,\alpha^\pi)$,
     \item If $\alpha$ is a boolean combination of formulas $\vartheta$, then $\alpha^\pi$ is the same boolean combination of formulas $\vartheta^\pi$,
     \item If $\alpha$ is $\exists x\, \vartheta$, then $\alpha^\pi$ is the formula $\exists \bar{x}\, \partial^\pi(\bar{x}) \wedge \vartheta^\pi$, and
     \item If $\alpha$ is $\forall x\, \vartheta$, then $\alpha^\pi$ is the formula $\forall \bar{x}\, \partial^\pi(\bar{x}) \to \vartheta^\pi$.
 \end{enumerate}
 
\end{definition}

\begin{definition}
  Suppose $\pi : L \to K$ is an interpretation and $A$ is a $K$-structure. Then the \emph{$\pi$-translation} $A^\pi$ is the $L$-structure with domain $\partial^\pi[A^k]$ with the denotation of $\lambda$ given by $\lambda^\pi[A]$, for each $\lambda \in L$. 
  
  (Note that even though the arity of $\lambda^\pi$ is $nk$ as a $K$-formula, it defines an $n$-ary relation over $A^\pi$, whose elements are $k$-tuples of $A$.)
\end{definition}

\paragraph{Functional signatures} In a very particular case (see \emph{successor expansions}, Definition \ref{successor expansions}) we will want to consider signatures with function symbols, and exactly once, we will want to define an interpretation $\pi : L \to K$ where $L$ has some function symbol $f(x_1,\dots,x_n)$. In this case $f^\pi(\vec{x}_1,\dots,\vec{x}_n)$ is a \emph{definition by cases}, where each case is a first-order $K$-formula, and the \emph{definiens} inside each case is a term. In a quantifier-free interpretation, each case must be a quantifier-free $K$-formula.

\paragraph{Injective interpretations} Usually an interpretation will also contain a first-order $K$-formula $eq^\pi(\bar{x},\bar{y})$ defining when we regard two $k$-tuples as equal. (For example, when interpreting rational numbers by pairs of integers, we say $(a,b) = (c,d) \iff ac-bd =0$.) In case $eq^\pi$ is simply equality of tuples (as above), $\pi$ is called an \emph{injective interpretation}. Here we do not deal with any interpretations with a nontrivial equivalence relation, and moreover the assumption of injectivity will be important in Section \ref{sec:Wellorder-invariance}. Therefore, it is convenient to drop the word ``injective'' and simply refer to interpretations.

\begin{lemma}[Fundamental property of interpretations]
Suppose that $\pi : L \to K$ is an elementary interpretation. Then for every $K$-structure $A$, every $n$-ary $L$-sentence $\varphi$, and every $\bar{x}_1,\dots,\bar{x}_n$ in the domain $\partial^\pi(A^k)$ of $A^\pi$,
$$ A \models \varphi^\pi (\bar{x}_1,\dots,\bar{x}_n) \iff A^\pi \models \varphi (\bar{x}_1,\dots,\bar{x}_n).$$
\end{lemma}

Note that on the left-hand side, $(\bar{x}_1,\dots,\bar{x}_n)$ is regarded as an $nk$-tuple of elements in $A$, and on the right-hand side, it is regarded as an $n$-tuple of elements in $A^\pi$.

There is a useful categorical perspective on interpretations that we do not go into here. Briefly, interpretations are morphisms in the category of signatures, and induce morphisms in the category of structures under the ``semantics functor'' which takes a signature $L$ to the family of $L$-structures. From this viewpoint, interpretations behave like a sort of categorical adjoint.

\begin{definition}
 Suppose that $\pi : L \to K$, $\mathcal{L}$ is a class of $L$-structures, and $\mathcal{K}$ is a class of $K$-structures. Then $\pi$ is a \emph{left total interpretation $\mathcal{K \to L}$} in case for every $A \in \mathcal{K}$, $A^\pi \in \mathcal{L}$.
\end{definition}

\begin{definition}
Let $\mathcal{K}$ be a family of $K$-structures, $\mathcal{L}$ a family of $L$-structures, and $Q$ be an $nk$-ary query over $\mathcal{K}$. Then we say $Q$ is \emph{basic $\mathcal{(L,P)}$-invariant definable} in case there exists a $k$-ary left total interpretation $\pi : \mathcal{K \to L}$ and an $n$-ary $\mathcal{(L,P)}$-invariant formula $\varphi$ such that for every $A \in \mathcal{K}$ and $\bar{x}_1,\dots,\bar{x}_n \in A^\pi$, 
$$ A \models Q(\bar{x}_1,\dots,\bar{x}_n) \iff A^\pi \models (\mathfrak{P}\vec{f})\,\varphi(\bar{x}_1,\dots,\bar{x}_n).$$
\end{definition}

\begin{definition}
With the same notation as above, we say that $Q$ is \emph{$\mathcal{(L,P)}$-invariant definable} in case it is expressible as a first-order combination of basic $\mathcal{(L,P)}$-invariant definable queries.

We say $Q$ is \emph{$\mathcal{(L,P)}$-invariant axiomatizable} in case it is expressible as a countable intersection of $\mathcal{(L,P)}$-invariant definable queries.
\end{definition}

\paragraph{A note on terminology}
``Definable'' and ``axiomatizable'' are analogous to ``basic elementary'' and ``elementary'' classes of structures according to the usual conventions of first-order model theory. On the other hand, finite model theory calls definable classes ``elementary'' and does not typically deal with axiomatizable classes. To avoid confusion, we use ``definable'' and ``axiomatizable,'' following Hodges \cite{Hod93}.

\fi

\section{Traversals}\label{sec:traversals}
Traversals are absolutely fundamental in computer science. They give us systematic ways of exploring a finite graph or other sort of network, and lie at the foundation of all sorts of sophisticated algorithms and techniques. 
Let us isolate the simplest possible version, which we call \emph{generic graph search}, and which operates over a finite nonempty graph $G$.
\begin{enumerate}
    \item Initialize a set $S$ to some vertex in $G$, and repeat the following until $G \setminus S$ is empty.
    \item If there is some vertex in the boundary of $S$, add it to $S$. Otherwise, add any element of $G \setminus S$ to $S$.
\end{enumerate}
Generic graph search is nondeterministic, insofar as it does not specify which vertex to add to $S$. Important refinements of this algorithm include \emph{breadth-first} and \emph{depth-first} search, which specify additional heuristics for how to add vertices to $S$, without being fully deterministic.

In common parlance, the word \emph{traversal} can refer either to the algorithm or the linear orders of $G$ they produce, but in the current work we reserve the term ``traversal,'' ``breadth-first traversal,'' and ``depth-first traversal'' for the latter.
In this paper, we do not work with depth-first traversals, but we will come back to them in the last section.

\begin{definition}
For a finite graph $G$, $(G,<)$ is a \emph{traversal} (resp. breadth-first traversal, depth-first traversal) in case some instance of generic graph search (resp. breadth-first search, depth-first search) of $G$ visits its vertices in order $<$. 
\end{definition}

Corneil and Krueger \cite{CoKr08} discovered that, in fact, these traversals are first-order definable in the language of ordered graphs.
\begin{lemma}
For any finite graph $G$,
\begin{align*}
 (G,<)\text{ is a traversal} &\iff (G,<) \models (\forall u < v < w)(uEw \to (\exists x<v) xEv), \\
 (G,<)\text{ is a breadth-first traversal} &\iff (G,<) \models (\forall u < v < w)(uEw \to (\exists x < v) x \le u \wedge xEv), \\
 (G,<)\text{ is a depth-first traversal} &\iff (G,<) \models (\forall u < v < w)(uEw \to (\exists x < v) x \ge u \wedge xEv).
\end{align*}
\end{lemma}



Note that connected components of $G$ induce intervals in a traversal. Notice also how the definitions of breadth-first traversal and depth-first traversal refine the notion of traversal in opposing ways: given a vertex $v$ that occurs between two endpoints $u$ and $w$ of a single edge, $v$ must have some prior neighbor in a plain traversal. In a breadth-first traversal, there must be some prior neighbor less than or equal to $u$, and in a depth-first traversal, there must be some prior neighbor greater than or equal to $u$.

It is an easy but very important fact that

\begin{lemma}
Every finite graph admits a breadth-first traversal and a depth-first traversal; a fortiori, every finite graph admits a traversal.
\end{lemma}

In the present paper we characterize L and NL using traversals and breadth-first traversals respectively; it is an open question whether depth-first traversals similarly characterize some complexity class.

\subsection{Traversal-invariant definability}
We now present the fundamental definability-theoretic concepts in this paper.  We use the standard model-theoretic notion of an \emph{interpretation} in this definition; for details see the Appendix. If $\mathcal{K}$ is some family of structures in a common signature, by a ``query over $\mathcal{K}$,'' we mean a boolean query, i.e., an isomorphism-closed subset of $\mathcal{K}$.\footnote{We will represent $n$-ary queries over $\mathcal{K}$ by boolean queries over the family of structures obtained by expanding every structure in $\mathcal{K}$ by any $n$ points.}

\begin{definition}
Suppose $K \subseteq K^+$ are signatures, $\mathcal{K}$ is a nonempty family of $K$-structures, and $\mathcal{P}$ is a nonempty family of $K^+$-structures, such that for any $A \in \mathcal{P}$, its $K$-reduct $A|_K$ is in $\mathcal{K}$. 

A first-order sentence $\varphi$ over $\mathcal{P}$ is \emph{$(\mathcal{K},\mathcal{P})$-invariant} in case for any two structures $A$ and $B$ in $\mathcal{P}$, if $A|_K \cong B|_K $, then $A \models \varphi \iff B \models \varphi$.
\end{definition}

\begin{definition}
  An \emph{$n$-pointed graph} is a graph expanded with $n$ constants. Let $\Gamma_n$ be the language of $n$-pointed graphs, i.e., a binary relation symbol and $n$ constant symbols.
\end{definition}

\begin{definition}
 Let $\mathcal{G'}$ be a family of finite $n$-pointed graphs, $\mathcal{T}$ be the set of all expansions of structures in $\mathcal{G'}$ by any traversal, and $\varphi$ be a $\mathcal{(G',T)}$-invariant sentence. Then for any $G \in \mathcal{G}'$, we write
 $$G \models (\mathfrak{T} <)\,\varphi$$
 to indicate that for some (equivalently, any) traversal $<$ of $G$, $(G,<) \models \varphi$. Similarly, we write
 $$G \models (\mathfrak{B} <)\,\varphi$$
 if $\varphi$ is $\mathcal{(G',B)}$-invariant where $\mathcal{B}$ is the set of all expansions by breadth-first traversals.
\end{definition}

\begin{definition}
Let $K$ be a signature, $\mathcal{K}$ be some family of $K$-structures and $Q$ a query over $\mathcal{K}$. We say that $Q$ is \emph{basic traversal-invariant definable} if there exists some $n\in\mathbb{N}$, a family of finite $n$-pointed graphs $\mathcal{G}'$, a $\mathcal{(G',T)}$-invariant sentence $\varphi$, and an interpretation $\pi : \Gamma_n \to K$, such that 
\begin{enumerate}
\item $\pi$ is left-total from $\mathcal{K}$ to $\mathcal{G'}$, and
\item for any $A \in \mathcal{K}$,
$$ A \in Q \iff A^\pi \models (\mathfrak{T} <)\,\varphi,$$
where $\mathcal{T}$ is the set of all expansions by a traversal of all graphs in $\mathcal{G}'$.\footnote{See the Appendix for the definition of notions and notations involving interpretations.}

\end{enumerate}
 We define \emph{basic breadth-first traversal (BFT)-invariant definable} similarly. We also write $A \models ((\mathfrak{T}<)\,\varphi)^\pi$ to mean $A^\pi \models (\mathfrak{T}<)\,\varphi$.
\end{definition}

Note that in our definition of traversal- or BFT-invariant definability, $\mathcal{G'}$ is not required to be the family of \emph{all} finite $n$-pointed graphs, though it typically will be. Note also that $\mathcal{K}$ must be a family of finite structures if there is to be a left total interpretation $\pi : \mathcal{K \to G'}$.

\begin{definition}
 Let $K$ be a signature, $\mathcal{K}$ be some family of $K$-structures and $Q$ a query over $\mathcal{K}$. Then $Q$ is \emph{traversal-invariant definable} (resp. \emph{BFT-invariant definable}) if it is a boolean combination of basic traversal-invariant (resp. basic BFT-invariant) definable queries.
\end{definition}

We collect some important examples:

\begin{lemma}
The following queries are traversal-invariant definable over the indicated families of structures $\mathcal{K}$:
\begin{enumerate}
    \item Undirected $st$-connectivity, over all finite 2-pointed graphs.
    \item The family of all acyclic graphs, over all finite graphs.
    \item The family of all bipartite graphs, over all finite graphs.
    \item The family of even-sized finite linear orders, over all finite linear orders.
\end{enumerate}
\end{lemma}

\begin{proof}
Let $\mathcal{G_2}$ be the family of all finite 2-pointed graphs with constants $s$ and $t$. The binary reachability relation is actually definable by a single $\mathcal{(G_2,T)}$-invariant sentence, which says that there is no $w$ with no prior neighbor such that $s < w \le t$ or $t < w \le s$. Since components of $G$ induce intervals in $(G,<)$, this formula asserts there is no interval separating $s$ and $t$ into separate connected components.

Acyclicity is similarly the spectrum of a $\mathcal{(G,T)}$-invariant sentences. A graph is acyclic iff, relative to any traversal, no vertex has two or more prior neighbors.

The \emph{square} of a graph $G = (V,E)$ is the graph $G^2 = (V,E^2)$, where $E^2(x,y)$ iff $x$ and $y$ are connected by a path of length exactly two. Then $G$ is bipartite iff $G^2$ is disconnected. Since $G^2$ is definable as a translation of $G$ under an interpretation $\pi : \mathcal{G \to G}$, and since connectivity is traversal-invariant definable, so is bipartiteness.

The parity of a linear order is also equivalent to the connectivity of a translation. Specifically, connect $u$ and $v$ by an edge iff $u = v \pm 2 \mod n$, where $n$ is the size of the order. Then the resulting graph is either a single cycle or a union of two cycles depending on whether $n$ is odd or even respectively.
\end{proof}

Since, e.g., connectivity and acyclicity are not Gaifman-local queries \cite{Gu84,Sch13}, it follows that traversal-invariance is strictly more expressive than order-invariance.

\subsection{Directed reachability}\label{subsec:Directed reachability}

Here we deal with the question of \emph{directed} $st$-connectivity using BFT invariance. In fact, we will need something more than the directed graph structure, but the result will be invariant of it, an apparent asymmetry with the undirected case that will be resolved in the next section.

This construction is substantially more sophisticated than our examples above. We reduce directed reachability to an equidistance problem over undirected graphs, which we solve with the appropriate BFT-invariant sentence.

\begin{definition}\label{successor expansions}
If $A$ is a finite structure, we say that a \emph{successor expansion} $(A,S)$ of $A$ is a structure of the form $(A,\min,\max,S)$, where $\min$ and $\max$ are constants and $S$ is a successor \emph{function} on a total order with endpoints $\min$ and $\max$. 

If $K$ is a signature, let $(K,S)$ be the signature of successor expansions of $K$-structures.\footnote{Note that there is no symbol for the order with respect to which $S$ is a successor function in the signature $(K,S)$.}
\end{definition}

\begin{definition}[Successor Invariance]
 For any family $\mathcal{K}$ of finite structures, let $\mathcal{K^S}$ be the set of all successor expansions of $\mathcal{K}$. A query $Q$ over $\mathcal{K^S}$ is \emph{successor-invariant} in case for any $A,B \in \mathcal{K^S}$, if $A|_K \cong B|_K$, then $A \models Q \iff B \models Q$.
 
 For any $C \in \mathcal{K}$, we say $C \models (\mathfrak{S} S)\,Q$ iff some (equivalently, any) successor expansion of $C$ satisfies $Q$.
\end{definition}

\begin{definition}
 Let $\mathcal{D}_n$ be the family of all finite $n$-pointed directed graphs, and $\mathcal{G}_n$ be the family of all finite $n$-pointed graphs. Let $\mathcal{D}_n^S$ be the family of all successor expansions of all finite directed $n$-pointed graphs.
\end{definition}

\paragraph{The interpretation $\rho$.} We present an interpretation defined in \cite{Tom91} that translates directed successor graphs into undirected graphs. Let $(x,y,z)$ be the constants of $\Gamma_3$ and $(s,t)$ be the constants of $(\Gamma_2,S)$. Consider the binary interpretation $\rho : \Gamma_3 \to (\Gamma_2,S)$ defined by 
$$ E^\rho(u,a;v,b) \equiv \big(S(a)=b \wedge E(u,v)\big) \vee \big(S(b)=a \wedge E(v,u) \big).$$
$$x^\rho = (s,\min) $$
$$y^\rho = (s,\max) $$
$$z^\rho = (t,\max) $$
Then $\rho$ is left total as an interpretation $\mathcal{D^S_\text{2} \to G_\text{3}}$, because the predicate $E^\rho$ is visibly symmetric. Note that $\rho$ is also quantifier-free. We can express the $st$ reachability problem on $D \in \mathcal{D}^S_2$ into an equidistance problem on $D^\rho$. 
A proof of the following lemma is presented in \cite{Tom91}, and reproduced in the Appendix.

\begin{lemma} \label{lem:rho-property}
For any graph $D \in \mathcal{D}^S_2$, there is a directed path from $s$ to $t$ in $D$ iff the vertices $y$ and $z$ are equidistant from $x$ in $D^\rho$. Even stronger, if there is no directed path from $s$ to $t$ in $D$, then either $d(x,y)$ or $d(x,z)$ is undefined or
$$|d(x,y) - d(x,z) | \ge 2,$$
where $d$ indicates distance in $D^\rho$.
\end{lemma}

\begin{definition}
Let $\mathcal{G}'_3$ be the family of finite 3-pointed undirected graphs $(G,x,y,z)$ such that if $x$, $y$, and $z$ are connected,
$$ |d(x,y) - d(x,z)| \neq 1 .$$
\end{definition}

By Lemma \ref{lem:rho-property}, $\rho$ is in fact a left-total interpretation $\mathcal{D}^S_2 \to \mathcal{G}'_3$.

\paragraph{Breadth-first traversals and quasi-levels}
On a graph with a distinguished source for each connected component, vertices are naturally partitioned into \emph{levels} according to their distance from their respective source. If we fix a BFT of a graph, and let the source of each connected component be its least element, then the resulting levels induce intervals in that traversal.\footnote{Recall that connected components induce intervals of a traversal, so it suffices to observe that levels induce intervals within connected components.} Moreover, every edge of the graph is either within levels or between adjacent levels. The least neighbor of every vertex (except the source) is in the previous level.

It is not clear that it is possible to define the property that two nodes are in the same level using a first-order formula on a graph expanded by a BFT. However we can do almost as well.

\begin{definition}
 Let $(V,E,<)$ be a finite graph expanded by a breadth-first traversal. A \emph{quasi-level} is a nonempty interval $I$ of $(V,E)$ such that $w \in I \iff p(w) < v \le w$, where $v$ is the least element of $I$ and $p(w)$ the least neighbor of $w$.
\end{definition}

For example, consider the following tree with the breadth-first ordering indicated in the subscript. Then $\{v_2,v_3\}$, $\{v_3,v_4,v_5\}$, $\{v_6,\dots,v_{11}\}$, and $\{v_{13},v_{14},v_{15}\}$ are quasi-levels but $\{v_5,\dots,v_{11}\}$ and $\{v_4,v_5,v_6\}$ are not. In the first counterexample, $v_{11} \in I$, but $p(v_{11})=v_5$, the least element of $I$, and in the second counterexample, $p(v_7) < v_4 < v_7$ and $v_4$ is the least element of $I$, but $v_7 \notin I$.

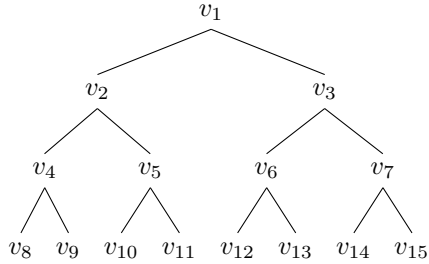
\begin{center}
\begin{tikzpicture}
\Tree
[.$v_1$
	[.$v_2$
		[.$v_4$
			[.$v_8$
			]
			[.$v_9$
			]
		]
		[.$v_5$
			[.$v_{10}$
			]
			[.$v_{11}$
			]
		]
	]
	[.$v_3$
		[.$v_6$
			[.$v_{12}$
			]
			[.$v_{13}$
			]
		]
		[.$v_7$
			[.$v_{14}$
			]
			[.$v_{15}$
			]
		]
	]
]
\end{tikzpicture}

\end{center}

Observe that it is easy to define when two vertices $v$ and $w$ occur in a common quasi-level, by the formula
$$ (p(w) < v \le w) \vee (p(v) < w \le v) .$$

Notice that if two vertices occur in a common quasi-level, then their distances from their (necessarily common) source cannot differ by more than 1. If two vertices occur in no common quasi-level, then either they are in different connected components, or the distances from their common source cannot be equal.

\paragraph{The interpretation $\tau$.}
We define a 2-dimensional interpretation $\tau : \Gamma_6 \to \Gamma_3$. Let $$(x_1,y_1,z_1,x_2,y_2,z_2)$$ be the constants in $\Gamma_6$ and $(x,y,z)$ be the constants in $\Gamma_3$. Given $G \in \mathcal{G}_3$, the domain of $G^\tau$ consists of ``two copies'' of $G$, which we achieve by $\partial^\tau (u,v) \iff v = x \vee v = y$. Within each copy, we inherit the edge relation from $G$, and we let, e.g., $x_i$ be the vertex corresponding to $x$ in copy $i$. We do not put any edges between the two copies except for connecting $x_1$ and $x_2$. Notice that $\tau$ is quantifier-free.

\begin{definition}
 Let $\mathcal{G}'_6$ be $\{ G^\tau : G \in \mathcal{G}'_3\}$. 
\end{definition} 

Then (by definition), $\tau$ is a left-total interpretation $\mathcal{G}'_6 \to \mathcal{G}'_3$. Moreover,

\begin{theorem}\label{formula-psi}

Let $\mathcal{B}$ be the set of all expansions of graphs in $\mathcal{G}'_6$ by a breadth-first traversal. There is a $(\mathcal{G}'_6,\mathcal{B})$-invariant formula $\psi$ such that for any $(G,x,y,z) \in \mathcal{G}'_3$,
$$d(x,y) = d(x,z) \implies G^\tau \models (\mathfrak{B} <)\,\psi $$
$$|d(x,y) - d(x,z)| \ge 2 \implies G^\tau \models \neg (\mathfrak{B} <)\, \psi,$$
where the second case also contains all those graphs where $x$, $y$, and $z$ are not all connected.
\end{theorem}

\noindent
(The proof is deferred to the Appendix.)

By composing the interpretation $\rho$ with the interpretation $\tau$, we see that for any successor expansion of a finite 2-pointed directed graph $D \in \mathcal{D}_2^S$, there is a path from $s$ to $t$ in $D$ if and only if $D^{\rho \tau} \models (\mathfrak{B} <)\, \psi$. Hence, 

\begin{corollary}\label{directed reachability is BFT inv}
 The directed reachability query is BFT-invariant definable over $\mathcal{D}^S_2$.
\end{corollary}

\section{Descriptive Complexity}\label{sec:descriptive complexity}
In this section we obtain the main results of this paper: a characterization of deterministic and nondeterministic logarithmic space by traversal and breadth-first traversal invariance quantifiers respectively.

\subsection{Multihead finite automata}

A \emph{nondeterministic multihead finite automaton} (NMFA) is an automaton with a single tape, finitely many heads on that tape, and a finite control. Unlike a Turing machine, the tape is not infinite; rather, it is initialized to the input string plus two special characters on either side to mark the left and right endpoints. Also unlike a Turing machine, the heads cannot write, they can only move left, right, or stay put depending on which characters they are reading.
A single state is designated as accepting; if the computation enters this state then we say it halts. The language of an NMFA is exactly the set of strings it halts on.

Formally, an NMFA consists of a set $Q$ of states, some number $k \in \mathbb{N}$ of heads, an input alphabet $\Sigma$, a start state $q_0 \in Q$, an accept state $q_f \in Q$, and a transition relation $\delta$ which relates $k$-tuples in $\Sigma \cup \{\triangleright,\triangleleft\}$ with $\{-1,0,1\}^k$. If any head is reading the left (respectively right) endpoint character, no subsequent transition may move that head right (respectively left). Futhermore, if the current state is $q_f$, then the transition relation moves all heads to the left (if possible) or fixes them if they are already at the left endpoint.

A \emph{configuration} of an NMFA consists of the input string, the current state, and the location of the heads. The transition relation induces a relation on the space of configurations in the natural way. The \emph{intial configuration} is the one in which the state is $q_0$ and all heads are at the left. The \emph{final configuration} is the same but with state $q_f$. By the stipulation of the transition relation, if an NMFA enters $q_f$, then it will always enter the final configuration.

The \emph{configuration graph} of an NMFA on a particular input $x$ is a 2-pointed directed graph whose vertices are the set of configurations on $x$ and whose edge relation is the graph of the relation induced by the transition function. The source and sink are the initial and final configurations respectively.

\paragraph{Strings and pointed graphs as structures}
Let $\Gamma_2$ be the language of 2-pointed graphs, and let $(\Gamma_2,S)$ be the language of 2-pointed successor graphs, with two (additonal) constants $\min$ and $\max$, and a successor function.

Let $\Sigma$ be a finite alphabet. We think of a string $x = x_0x_1\dots x_{n-1}$ in $ \Sigma^\star$ as a finite structure with domain $\{0,1,\dots,n-1\}$, a predicate $\sigma$ for each $\sigma \in \Sigma$ with semantics $$(\forall i < n)\ x \models \sigma(i) \iff x_i = \sigma,$$
constants min and max naming 0 and $n-1$, and finally a successor function taking index $i$ to index $i+1$.\footnote{This is a common, ``folklore,'' method of representing strings as structures. Often one takes a total ordering $<$ over the indices of a string instead of the successor function (see Libkin \cite{Lib04}), but for our purposes, either will work.}

We henceforth overload the meaning of $\Sigma$ to indicate not only an alphabet, but also the signature of strings in that alphabet, so that the terms ``finite $\Sigma$-structure'' and ``member of $\Sigma^\star$'' denote the same objects.

Crucial to our work is that for a fixed NMFA, there is an interpretation that takes an input string and translates it into the associated configuration graph.

\begin{theorem}\label{config graph-interpretation}
For every NMFA $\mathcal{M}$ with alphabet $\Sigma$ there is an interpretation $\pi : \Gamma_2 \to \Sigma$ such that for every sufficiently long string $x \in \Sigma^\star$, $x^\pi$ is isomorphic to the configuration graph of $\mathcal{M}$ on input $x$.

Furthermore, we can expand $\pi$ to an interpretation $\pi : (\Gamma_2,S) \to \Sigma$, so that $x^\pi$ is a successor expansion of the above configuration graph.

Moreover, $\pi$ can be made quantifier-free.
\end{theorem}

\noindent
(The proof is deferred to the Appendix)

\begin{definition}
 An NMFA $\mathcal{M}$ is \emph{symmetric} (SMFA) in case, for any input $x$, the configuration graph of $\mathcal{M}$ on $x$ is undirected.
\end{definition}

Computability by NMFAs is known to capture exactly nondeterministic logarithmic space (NL) \cite{Hartmanis1972}, and computability by SMFAs captures at least logarithmic space (L). (In \cite{Axe12}, Axelsen considers the more restrictive \emph{reversible} MFAs, which are both deterministic and symmetric, and shows that they capture L.)

\subsection{Capturing L and NL}

\paragraph{Canonical encodings}
For any finite structure $A$, any linear order $(A,<)$, and any fixed alphabet $\Sigma$ of size at least 2, there is a \emph{canonical encoding} of $(A,<)$ as a string in $\Sigma^\star$.

This construction is the foundation of all results in descriptive complexity, and can be found in numerous texts, e.g., \cite{Lib04}. We will not repeat it here. We do note, however, that any successor expansion $(A,S)$ of $A$ induces a linear order---hence every successor expansion of any finite structure has a canonical encoding.

Even more importantly, this canonical encoding is \emph{definable} as the translation induced by a quantifier-free interpretation. The details are complicated, but can be found in Section 9.2 of \cite{Lib04}.

\begin{theorem}\label{encoding-interpretation}
For every signature $L$, there is a quantifier-free interpretation $\mu: \Sigma \to (L,S)$ such that for every successor expansion $(A,S)$ of any finite $L$-structure $A$, $(A,S)^\mu$ is the canonical encoding of $(A,S)$.
\end{theorem}

We now state the definition of a complexity-bounded query over finite structures, for which we need to imagine models of computation that take finite structures as input. We follow the standard method in descriptive complexity, which is to take a model of computation that operates on strings, and feed it the encoding $(A,S)^\mu$ of a structure $A$. Of course this encoding is not canonical given only $A$; for a well-defined query, we demand that the result of the computation is invariant of the particular expansion $(A,S)$.

Now we are in a position to state:

\begin{theorem}\label{thm:logspace is TI}
For every signature $L$ and every logarithmic space query $Q$ over finite $L$-structures, there is a quantifier-free interpretation $\gamma : \Gamma_2 \to (L,S)$ such that for every sufficiently large finite $L$-structure $A$, $$ A \in Q \iff A \models (\mathfrak{S}S)\, ((\mathfrak{T} <)\, \varphi)^\gamma,$$
where $(\mathfrak{T} <)\, \varphi$ is the sentence in the language of 2-pointed ordered graphs asserting that the distinguished vertices are connected.
\end{theorem}

\begin{proof}
Let $\mu:\Sigma \to (L,S)$ be the interpretation given 
in Theorem \ref{encoding-interpretation}, let $\mathcal{M}$ be an SMFA deciding $Q$, let $\pi : \Gamma_2 \to \Sigma$ be the associated interpretation from Theorem \ref{config graph-interpretation}, and let $\gamma = \mu \pi$. Fix a finite $L$-structure $A$ and an arbitrary successor expansion $(A,S)$.

Then $\mathcal{M}$ accepts the string $(A,S)^\mu$ iff $A \in Q$. But, $(A,S)^{\mu \pi} = (A,S)^\gamma$ is the configuration graph of $\mathcal{M}$ on $(A,S)^\mu$, so $\mathcal{M}$ accepts $(A,S)^\mu$ just in case the distinguished vertices of $(A,S)^\gamma$ are connected. In other words,
$$ A \in Q \iff (A,S)^\gamma \models  (\mathfrak{T} <)\, \varphi.$$ 
Therefore,
$$ A \in Q \iff (A,S) \models ((\mathfrak{T} <)\, \varphi)^\gamma ,$$
and since the right hand side is independent of the particular successor expansion,
$$ A \in Q \iff A \models (\mathfrak{S}S)\, ((\mathfrak{T} <)\, \varphi)^\gamma. \qedhere$$
\end{proof}

\begin{theorem}
For every signature $L$ and every \emph{NL} query $Q$ over finite $L$-structures, there is a quantifier-free interpretation $\gamma : \Gamma_6 \to (L,S)$ such that for every sufficiently large finite $L$-structure $A$, $$ A \in Q \iff A \models (\mathfrak{S}S)\, ((\mathfrak{B} <)\, \psi)^\gamma,$$
where $\psi$ is the sentence from Theorem \ref{formula-psi} in the language $(\Gamma_6,<)$.
\end{theorem}

\begin{proof}
Let $\mu:\Sigma \to (L,S)$ be the interpretation given 
in Theorem \ref{encoding-interpretation}, let $\mathcal{M}$ be an NMFA deciding $Q$, let $\pi : (\Gamma_2,S) \to \Sigma$ be the associated interpretation from Theorem \ref{config graph-interpretation}. Recall the interpretations $\rho:\Gamma_3 \to (\Gamma_2,S)$  and $\tau:\Gamma_6 \to \Gamma_3$ from Section \ref{subsec:Directed reachability}. Finally, let $\gamma = \mu \pi \rho \tau$. Fix a finite $L$-structure $A$ and an arbitrary successor expansion $(A,S)$.

Then $\mathcal{M}$ accepts the string $(A,S)^\mu$ iff $A \in Q$. But $\mathcal{M}$ accepts $(A,S)^\mu$ just in case there is a path from source to sink over the graph $(A,S)^{\mu \pi}$. By Corollary \ref{directed reachability is BFT inv}, this occurs just in case $(A,S)^{\mu \pi \rho \tau} \models (\mathfrak{B} <)\,\psi$. But $(A,S)^{\mu \pi \rho \tau} = (A,S)^\gamma$.

Since the above is independent of the particular successor expansion $S$, 
$$ A \in Q \iff A \models (\mathfrak{S} S)((\mathfrak{B} <)\,\psi)^\gamma,$$ which completes the proof.
\end{proof}

\subsection{Logspace-computable traversals}
In the other direction, we want to show that traversals and breadth-first traversals are computable in L and NL respectively. These constructions rely on the computability in logarithmic space of undirected $st$-connectivity, and furthermore on the existence of logarithmic space \emph{universal exploration sequences} \cite{Rei08}.

Given an ordered graph $G$ and a vertex $v$, it is possible to construct, in logarithmic space, the index of the least vertex $u$ in the connected component of $v$. Simply iterate through the vertices of $G$ in order, testing connectivity with $v$, until we find a vertex that is connected.

\begin{theorem}\label{thm:logspace computable traversal}
There is a logarithmic space Turing machine which, for every finite ordered graph $(G,<)$, computes a traversal $(G,\prec)$ in the following sense: given the canonical encoding of $(G,<)$ and (indices of) two of its vertices $v$ and $w$, accepts or rejects according to whether $v \prec w$.
\end{theorem}

\begin{proof}
Given two vertices $v$ and $w$ in $G$, first test whether they are in the same connected component. If not, let $v_0$ and $w_0$ be the least elements in the connected components of $v$ and $w$ respectively, and compare $v$ and $w$ according to whether $v_0 < w_0$.

Otherwise, let $v_0$ be the least element of their common connected component. If $n = |G|$, construct (using space logarithmic in $n$) a universal exploration sequence, and explore the connected component of $v$ and $w$ according to that sequence starting with $v_0$. Let $v \prec w$ iff the first occurrence of $v$ precedes the first occurrence of $w$.

We must show $(G,\prec)$ is a traversal. Notice that connected components induce intervals.  If $v$ is not the least vertex in some connected component, then its first occurrence in the universal exploration sequence has some immediate predecessor $u$ which is a neighbor. Therefore, in the traversal, $u\prec v $; hence, $v$ has some preceding neighbor. 
\end{proof}

\paragraph{Canonical BFT of an ordered graph}
Unlike the case of ordinary traversals, where the traversal $(G,\prec)$ of $(G,<)$ depends on some family of universal exploration sequences, we will define a canonical breadth-first traversal $(G,\prec_B)$ of an ordered graph $(G,<)$ and show that it can be computed in nondeterminstic logspace.

\begin{definition}
Given a finite ordered graph $(G,<)$ and vertices $v,w \in G$, let $v_0$ and $w_0$ be the $<$-least elements of the connected components of $v$ and $w$ respectively. Let $<^\star$ be the ordering on finite sequences of vertices that orders them \emph{first} by length, and \emph{then} lexicographically. Let $\vec{v}$ be the $<^\star$-least path from $v_0$ to $v$.
Then,
\begin{enumerate}
    \item if $v_0 \neq w_0$, then $v \prec_B w \iff v_0 < w_0$,
    \item if $v_0 = w_0$ then $v \prec_B w \iff \vec{v} <^\star \vec{w}$.
\end{enumerate}
\end{definition}

\begin{lemma}\label{lem:canonical-BFT}
For any finite ordered graph $(G,<)$, $(G,\prec_B)$ is a breadth-first traversal.
\end{lemma}

\noindent
(Proof deferred to appendix.)

\begin{theorem}\label{lem:nlogspace-computable-BFT}
There is a logarithmic space nondeterministic Turing machine which, on input a finite ordered graph $(G,<)$ and vertices $v,w \in G$, decides whether or not $v \prec_B w$.
\end{theorem}

\begin{proof}
As in the proof of Theorem \ref{thm:logspace computable traversal}, given two vertices $v$ and $w$, first test whether they are in the same connected component. 
If not, decide $v \prec_B w$ according to whether $v_0 < w_0$.

Otherwise we argue that we can construct the sequence $\vec{v} = (v_0,\dots,v_{\ell-1},v)$ in the following sense: given an index for $v_i$, we can test whether it's equal to $v$; if not, we can construct the index of $v_{i+1}$, all in logarithmic space.

If we can do this, then we decide $v \prec_B w$ by comparing $\vec{v} <^\star \vec{w}$. First we compare their lengths: we simultaneously construct $(v_{i+1},w_{i+1})$ from $(v_i,w_i)$, until the first index is $v$ or the second is $w$. Unless this happens at the same stage, we are done. (Since $(v_{i+1},w_{i+1})$ overwrites $(v_i,w_i)$, this remains in logarithmic space.)

Otherwise, we start over, and simultaneously construct $(v_i,w_i)$ until we (necessarily) find the first index at which they differ. Then we decide $v \prec_B w$ according to which is larger.

It remains to show how to construct $v_{i+1}$ from $v_i$. Orient all edges in $G$ so that they increase distance from $v_0$. Then $v_{i+1}$ is the $<$-least vertex $x$ such that there is an edge $(v_i,x)$ and a directed path $(x,v)$. Since we can compute directed reachability in nondeterminstic logarithmic space, we can find $v_{i+1}$ in nondeterminstic logarithmic space as well.
\end{proof}

At this point we are ready to state the two central results of this paper.
\begin{theorem}\label{main-1}
The following are equivalent:
\begin{enumerate}
    \item $Q$ is a logspace-decidable query over finite $K$-structures.
    \item There is a quantifier-free interpretation $\pi:\Gamma_2 \to (K,S)$ such that for all sufficiently large finite $K$-structures $A$, $$ A \in Q \iff A \models (\mathfrak{S}S)\, ((\mathfrak{T} <)\, \varphi)^\pi, $$ where $\varphi$ is the formula expressing undirected $st$-connectivity.
    \item There is a traversal-invariant definable query $R$ over finite $(K,S)$ structures such that for any finite $K$-structure $A$, $$ A \in Q \iff A \models (\mathfrak{S}S)\, R $$
\end{enumerate}
\end{theorem}

\begin{proof}
Implication $1 \Rightarrow 2$ is exactly Theorem \ref{thm:logspace is TI}. Implication $2 \Rightarrow 3$ is immediate, as $((\mathfrak{T}<)\,\varphi)^\pi$ is by definition a traversal-invariant definable query, and traversal-invariant queries are closed under finite differences. It remains to show $3 \Rightarrow 1$.

It suffices to show that the traversal-invariant definable query $R$ is logspace computable over finite $(K,S)$ structures, as given an encoding of a $K$-structure $A$, a logspace Turing machine can always compute a successor relation on the domain of $A$, by using the particular encoding in which $A$ is presented.

Since logspace-computable queries are closed under boolean combinations, it suffices to show that any basic traversal-invariant definable query is logspace computable. Since logspace computable queries are closed under elementary interpretations, it suffices to show that for any class $\mathcal{G}'$ of finite graphs, every $(\mathcal{G}',\mathcal{T}')$-invariant formula is logspace computable over graphs in $\mathcal{G'}$, where $\mathcal{T'}$ is the family of all expansions of graphs in $\mathcal{G'}$ by traversals. 

But for this, it suffices to show that any first-order sentence in the language of ordered graphs is logspace 
computable given an encoding of a finite graph, where the order is the traversal defined in Theorem \ref{thm:logspace computable traversal}. Since logspace queries are closed under first-order combinations, it suffices to check that given any encoding of a graph and two vertices therein, we can test whether they are equal, test whether they are connected by an edge, or compare them according to the canonical traversal.

The first two are true, and the last is exactly Theorem \ref{thm:logspace computable traversal}.
\end{proof}
By replacing `L' by `NL' and `traversal' by `breadth-first traversal' throughout, we get

\begin{theorem}\label{main-2}
The following are equivalent:
\begin{enumerate}
    \item $Q$ is an nlogspace-decidable query over finite $K$-structures.
    \item There is a quantifier-free interpretation $\pi:\Gamma_6 \to (K,S)$ such that for all sufficiently large finite $K$-structures $A$, $$ A \in Q \iff A \models (\mathfrak{S}S)\, ((\mathfrak{B} <)\, \psi)^\gamma, $$ where $\psi$ is the sentence of Theorem \ref{formula-psi}.
    \item There is a breadth-first traversal-invariant definable query $R$ over finite $(K,S)$ structures such that for any finite $K$-structure $A$, $$ A \in Q \iff A \models (\mathfrak{S}S)\, R $$
\end{enumerate}
\end{theorem}

\paragraph{Structures with successor}
Suppose the signature $K$ contains the unary function symbol $S$, and that $\mathcal{K}$ is a family of finite $K$-structures in which $S$ is interpreted by a successor function. Then the quantifier $(\mathfrak{S} S)$ in any successor-invariant query $(\mathfrak{S} S)\,R$ is superfluous over $\mathcal{K}$; i.e., for any $A \in \mathcal{K}$,
$$ A \models R \iff A \models (\mathfrak{S} S)\, R. $$
The reason is that the interpretation of $S$ in $R$ is independent of the particular successor function on $A$ that we choose, so we might as well choose the one native to $A$.

In particular, for such families $\mathcal{K}$, we can drop the $(\mathfrak{S} S)$ quantifier from the traversal- or breadth-first traversal-invariant queries $R$ of Theorems \ref{main-1} and \ref{main-2}. In particular, let us take the case of strings over a finite alphabet $\Sigma$, which are the original setting for logspace and nlogspace queries, and also successor structures as described in Section \ref{sec:descriptive complexity}. Then we have

\begin{corollary}
 For any family $Q \subseteq \Sigma^\star$,
 \begin{enumerate}
     \item $Q$ is logspace-decidable iff there is a traversal-invariant definable query $R$ such that for every string $x \in \Sigma^\star$, $x \in Q \iff x \models R$, and
     \item $Q$ is nlogspace-decidable iff there is a breadth-first traversal-invariant definable query $R$ such that for every string $x \in \Sigma^\star$, $x \in Q \iff x \models R$.
 \end{enumerate}
\end{corollary}

\subsection{Discussion and open questions} Our results are the first presentation-invariant characterizations of L and NL, and, to our knowledge, the largest known complexity classes characterized by first-order logic extended by invariant definability of an elementary class of presentations. They demonstrate the surprising power of interpretations (even quantifier-free ones!) and establish a new foundational correspondence between graph traversals and complexity classes.

The elephant in the room is whether \emph{depth-first traversal invariance} captures a meaningful complexity class, like polynomial time. While we have been able to find depth-first invariant definitions of certain suggestive queries (like \emph{vertex-avoiding paths}), we still do not have very strong evidence one way or the other. More generally, there are a variety of graph traversals and a variety of associated presentations (such as the ancestral relation of the traversal tree) which might correspond to interesting complexity classes.

Finally, we have extended these notions of definability to arbitrary infinite structures by requiring that the underlying order be well-founded. (Since well-orders are not elementarily definable, this circumvents the usual ``Beth definability'' obstacle to studying presentation invariance over infinite structures.) Whereas separating traversal-invariant from BFT-invariant definability over classes of finite structures requires separating L and NL, it is plausibly easier to separate them over arbitrary classes, and it is plausible that this will inform the finite case. This work is ongoing.


\clearpage

\appendix

\section{Interpretations and change of signature}

We review the basic definitions behind interpretations, following the exposition of Hodges \cite{Hod93}, except that we also allow for functional signatures (see below). There is no new mathematical content here; however, getting the definitions and terminology straight is terribly important, since we use interpretations extensively. 

\begin{definition}
 Let $L$ and $K$ be signatures and $k \in \mathbb{N}$. An \emph{elementary $k$-ary interpretation} $\pi : L \to K$ is a first-order $K$-formula $\partial^\pi(\bar{x})$, for each constant symbol $c \in L$ a variable-free $K$-term $c^\pi$, and for each relation symbol $r \in L$, a first-order $K$-formula $r^\pi(\bar{x}_1,\dots,\bar{x}_n)$, where $n$ is the arity of $r$, and the length of each tuple throughout is $k$.
 
 An interpretation is \emph{quantifier-free} in case $\partial^\pi$ and each $r^\pi$ is quantifier-free.
\end{definition}

\begin{definition}
 Given an elementary $k$-ary interpretation $\pi : L \to K$ and a first-order $L$-term or $L$-formula $\alpha$, its \emph{translation} $\alpha^\pi$, is a $K$-formula given by the following recursion:
 \begin{enumerate}
    \item If $\alpha$ is a variable $x$, then $\alpha^\pi$ is a $k$-tuple of (distinct) variables $\bar{x}$.
    \item If $\alpha$ is a constant symbol $c$, then $\alpha^\pi$ is $c^\pi$. 
     \item If $\alpha$ is the atomic formula $r(\alpha_1,\dots,\alpha_n)$, then $\alpha^\pi$ is $r^\pi(\alpha^\pi,\dots,\alpha^\pi)$,
     \item If $\alpha$ is a boolean combination of formulas $\vartheta$, then $\alpha^\pi$ is the same boolean combination of formulas $\vartheta^\pi$,
     \item If $\alpha$ is $\exists x\, \vartheta$, then $\alpha^\pi$ is the formula $\exists \bar{x}\, \partial^\pi(\bar{x}) \wedge \vartheta^\pi$, and
     \item If $\alpha$ is $\forall x\, \vartheta$, then $\alpha^\pi$ is the formula $\forall \bar{x}\, \partial^\pi(\bar{x}) \to \vartheta^\pi$.
 \end{enumerate}
 
\end{definition}

In the definition below, $\partial^\pi[A^k]$ denotes the subset of $A^k$ on which $\partial^\pi$ holds.

\begin{definition}
  Suppose $\pi : L \to K$ is an interpretation and $A$ is a $K$-structure. Then the \emph{$\pi$-translation} $A^\pi$ is the $L$-structure with domain $\partial^\pi[A^k]$ with the denotation of $\lambda$ given by $\lambda^\pi$, for each $\lambda \in L$. 
  
  (Note that even though the arity of $\lambda^\pi$ is $nk$ as a $K$-formula, it defines an $n$-ary relation over $A^\pi$, whose elements are $k$-tuples of $A$.)
\end{definition}

\paragraph{Functional signatures} In a very particular case (see \emph{successor expansions}, Definition \ref{successor expansions}) we will want to consider signatures with function symbols, and exactly once (Theorem \ref{config graph-interpretation}), we will want to define an interpretation $\pi : L \to K$ where $L$ has some function symbol $f(x_1,\dots,x_n)$. In this case $f^\pi(\vec{x}_1,\dots,\vec{x}_n)$ is a \emph{definition by cases}, where each case is a first-order $K$-formula, and the \emph{definiens} inside each case is a $k$-tuple of $K$-terms in the free variables $(\vec{x}_1,\dots,\vec{x}_n)$, where $k$ is the arity of $\pi$. In a quantifier-free interpretation, each case must be a quantifier-free $K$-formula.

Functional signatures also generalize signatures with constants, which are nullary function symbols. For a constant symbol $c$, $c^\pi$ is a definition by cases, where each case is a $k$-tuple of variable-free $K$-terms.

It is common practice in finite model theory to replace functions by their graph relations, thus working with purely relational signatures. The only reason for considering functional signatures here is to make certain interpretations quantifier free (cf. Theorems \ref{main-1} and \ref{main-2}); in the purely relational setting, these interpretations would contain quantifiers.

\paragraph{Injective interpretations} Usually an interpretation will also contain a first-order $K$-formula $eq^\pi(\bar{x},\bar{y})$ defining when we regard two $k$-tuples as equal. (For example, when interpreting rational numbers by pairs of integers, we say $(a,b) = (c,d) \iff ac-bd =0$.) In case $eq^\pi$ is simply equality of tuples (as above), $\pi$ is called an \emph{injective interpretation}. Here we do not deal with any interpretations with a nontrivial equivalence relation. Therefore, it is convenient to drop the word ``injective'' and simply refer to interpretations.

\begin{lemma}[Fundamental property of interpretations]
Suppose that $\pi : L \to K$ is an elementary interpretation. Then for every $K$-structure $A$, every $n$-ary $L$-sentence $\varphi$, and every $\bar{x}_1,\dots,\bar{x}_n$ in the domain $\partial^\pi[A^k]$ of $A^\pi$,
$$ A \models \varphi^\pi (\bar{x}_1,\dots,\bar{x}_n) \iff A^\pi \models \varphi (\bar{x}_1,\dots,\bar{x}_n).$$
\end{lemma}

Note that on the left-hand side, $(\bar{x}_1,\dots,\bar{x}_n)$ is regarded as an $nk$-tuple of elements in $A$, and on the right-hand side, it is regarded as an $n$-tuple of elements in $A^\pi$.

\begin{definition}
 Suppose that $\pi : L \to K$, $\mathcal{L}$ is a class of $L$-structures, and $\mathcal{K}$ is a class of $K$-structures. Then $\pi$ is a \emph{left total interpretation $\mathcal{K \to L}$} in case for every $A \in \mathcal{K}$, $A^\pi \in \mathcal{L}$.
\end{definition}

\subsection*{Properties of interpretations $\rho$ and $\tau$}

\paragraph*{Lemma \ref{lem:rho-property}}
For any graph $D \in \mathcal{D}^S_2$, there is a directed path from $s$ to $t$ in $D$ iff the vertices $y$ and $z$ are equidistant from $x$ in $D^\rho$. Even stronger, if there is no directed path from $s$ to $t$ in $D$, then either $d(x,y)$ or $d(x,z)$ is undefined or
$|d(x,y) - d(x,z) | \ge 2,$
where $d$ indicates distance in $D^\rho$.

\begin{proof}[Proof of Lemma \ref{lem:rho-property}]
(Adapted from \cite{Tom91}) Fix a graph $D$ and let $n$ be the number of vertices in $D$. Identify the vertices of $D$ with $\{0,1,\dots,n-1\}$ such that $S(i,i+1)$. Then in $D^\rho$, there is a path 
$$x=(s,0) - (s,1) - \dots - (s,n-1)=y,$$
of length $n-1$, and this is moreover the distance between $x$ and $y$, by considering the second coordinate.

If $t$ is reachable from $s$ in $D$, then that must be witnessed by some directed path $(s=r_0\to r_1 \to \dots\to r_{\ell-1}=t)$ of length $\ell \le n$. Then
$$(r_0,0) - (r_1,1) - \dots - (r_{\ell-1},\ell-1) - (r_{\ell-1},\ell) - \dots - (r_{\ell-1},n-1)$$
is a path in $D^\rho$ from $x$ to $z$ of length exactly $n-1$. Again by considering the second coordinate, we can see that there is no shorter path. Hence $y$ and $z$ are equidistant from $x$.

Conversely, suppose that there were a path in $D^\rho$ from $x$ to $z$ in $D^\rho$ of length exactly $n-1$. Then it must be of the form
$$ (u_0,0) - (u_1,1) - \dots - (u_{n-1},n-1),$$ where $u_0 = s$, $u_{n-1} = t$, and for each $i$, either $u_i =u_{i+1}$ or  $u_i \to u_{i+1}$ in $D$. Hence the $u_i$ witness a directed path from $s$ to $t$. 

Moreover, observe the parity of the second coordinate in any path must alternate. Hence, the length of any path from $(s,0)$ to $(t,n-1)$ must be equal to $n$ modulo 2. Therefore, if there is no directed path from $s$ to $t$ in $D$, then in $D^\rho$ then any path from $x$ to $z$ in $D^\rho$ must have length at least $n+2$.
\end{proof}

\paragraph{Theorem \ref{formula-psi}}
Let $\mathcal{B}$ be the set of all expansions of graphs in $\mathcal{G}'_6$ by a breadth-first traversal. There is a $(\mathcal{G}'_6,\mathcal{B})$-invariant formula $\psi$ such that for any $(G,x,y,z) \in \mathcal{G}'_3$,
$$d(x,y) = d(x,z) \implies G^\tau \models (\mathfrak{B} <)\,\psi $$
$$|d(x,y) - d(x,z)| \ge 2 \implies G^\tau \models \neg (\mathfrak{B} <)\, \psi,$$
where the second case also contains all those graphs where $x$, $y$, and $z$ are not all connected.

\begin{proof}[Proof of Theorem \ref{formula-psi}]
Let $\psi$ assert that all six constants $(x_1,\dots,z_2)$ occur in the same connected component; moreover, if $x_1<x_2$, then $y_2$ and $z_2$ occur in the same quasi-level, and if $x_2 < x_1$, then $y_1$ and $z_1$ occur in the same quasi-level. (To show that $\psi$ is invariant, it suffices to show that $\psi$ is correct.)

Fix a graph $(G,x,y,z) \in \mathcal{G}'_3$ such that $d(x,y)=d(x,z)$. Consider its translation, and expand this by an arbitrary breadth-first traversal. We may assume that all constants $(x_1,\dots,z_2)$ lie in the same connected component; otherwise $\psi$ evaluates to false, which is correct as not all of $(x,y,z)$ are connected.

Suppose that $x_1 < x_2$. Let $w$ be the least element of $<$ in the connected component of $x_1$. Then $w$ must be in $G_1$, so any path from $w$ to $y_2$ or $z_2$ must pass through the edge $(x_1,x_2)$. Hence, 
$$ |d(x,y) - d(x,z) | = |d(x_2,y_2) - d(x_2,z_2) | = |d(w,y_2) - d(w,z_2) |. $$
In other words, the desired quantity is exactly the difference in distance between $y_2$ and $x_2$ to the source. We know that this difference is either equal to 0 or at least 2, and $\psi$ correctly distinguishes these two cases by testing whether $y_2$ and $z_2$ occur in the same quasi-level.

Similarly, if $x_2 < x_1$, $\psi$ distinguishes $|d(x,y) - d(x,z)| = 0 $ from $|d(x,y) - d(x,z)| \ge 2$ by testing whether $y_1$ and $z_1$ occur in the same quasi-level.

\end{proof}

\section{Defining configuration graphs by interpretation}

\paragraph{Theorem \ref{config graph-interpretation}}
For every NMFA $\mathcal{M}$ with alphabet $\Sigma$ there is an interpretation $\pi : \Gamma_2 \to \Sigma$ such that for every sufficiently long string $x \in \Sigma^\star$, $x^\pi$ is isomorphic to the configuration graph of $\mathcal{M}$ on input $x$.
Furthermore, we can expand $\pi$ to an interpretation $\pi : (\Gamma_2,S) \to \Sigma$, so that $x^\pi$ is a successor expansion of the above configuration graph.
Moreover, $\pi$ can be made quantifier-free.

\begin{proof}[Proof of Theorem \ref{config graph-interpretation}]
Let $k$ be the number of heads in $\mathcal{M}$. Then $k+1$ will be the dimension of $\pi$. The domain of the interpretation $\partial_\pi$ just stipulates that the first coordinate is less than $q$, the number of states of $\mathcal{M}$.

We can now establish a bijection between the domain of $\pi$ and configurations of $\mathcal{M}$ with input $x$, for any string $x$ such that $|x| \ge q$. A configuration is simply specified by current state and the location of the heads, which correspond to the first and remaining $k$ coordinates of the domain respectively. Since $x$ is sufficiently long, there are enough choices in the first coordinate for all states of $\mathcal{M}$.

Let $\vec{u}$ and $\vec{v}$ be arbitrary configurations of $\mathcal{M}$ on input an arbitrary string of length at least $q$. We want to define $E^\pi(\vec{u},\vec{v})$ to hold just in case the configuration $\vec{v}$ is reachable from $\vec{u}$ in one step. This is definable by a boolean combination of formulas of the following form:
\begin{enumerate}
    \item $u_i$ is the minimum or maximum index,
    \item $\sigma \in \Sigma$ is the character at index, and
    \item indices $u_i$ and $v_i$ are identical or adjacent.
\end{enumerate}
Each of these formulas is quantifier-free definable in the language $\Sigma$, by e.g.,
\begin{enumerate}
    \item $u_i = 0$ or $u_i = n-1$,
    \item $\sigma(u_i)$, and
    \item $u_i = v_i$ or $S(u_i) = v_i$ or $u_i = S(v_i)$
\end{enumerate}
respectively, where $0$ and $n-1$ are aliases for $\min$ and $\max$ respectively. 
Finally, 
$$ s^\pi = (0,0,\dots,0),\  t^\pi = (S(0),0,\dots,0).$$
This is because all heads are at the left in the initial or final configuration, 0 is the start state, and 1 is the halt state.

Now for any string $x$ of length at least $q$, not only is the domain of $x^\pi$ in bijection with the configurations of $\mathcal{M}$ on $x$, but relative to this bijection $\pi_E$ defines the graph of ``next,'' and $\pi_s$ and $\pi_t$ are the initial and final configurations respectively. Hence $x^\pi$ as a structure is isomorphic to the configuration graph of $\mathcal{M}$ on input $x$.

To expand $\pi$ to an interpretation from $\Gamma^S_2$, we need to define a successor function on $(k+1)$-tuples of indices, given a successor function on indices. This is easy to do by mimicking the standard ``increment-by-one'' algorithm on numbers written in some fixed radix.\footnote{This is the only point in which we have to define $S^\pi$ where $S$ is a \emph{function symbol}, here we use a definition by cases in which every case is a quantifier-free term guarded by a quantifier-free formula.}
\end{proof}

\subsection*{First-order definitions of traversals}
\paragraph{Lemma \ref{lem:canonical-BFT}}
For any finite ordered graph $(G,<)$, $(G,\prec_B)$ is a breadth-first traversal.

\begin{proof}[Proof of Lemma \ref{lem:canonical-BFT}]
Connected components of $G$ induce intervals of $(G,\prec_B)$, so it suffices to assume that $G$ is connected. Let $v_0$ be the least element of $G$ (unambiguously with respect to either order).

Suppose $v$ is a non-minimal vertex and let $\vec{v} = (v_0,v_1,\dots,v_{\ell-1},v)$. Since $<^\star$-least shortest paths are closed under prefixes, $v_i \prec_B v$ for each $v$; in particular, $|\vec{v}_{\ell-1}| = \ell -1$.

Let $u$ be the $\prec_B$-least neighbor of $v$, and $\vec{u}$. Since $u \preceq_B v_{\ell-1}$, $|\vec{u}| \le \ell-1$. Since $(\vec{u},v)$ is a path from $v_0$ to $v$ of length at most $\ell$, and since $\ell$ is the distance from $v_0$ to $v$, $|\vec{u}|=\ell-1$.

Since $u$ and $v_{\ell-1}$ are the same distance from $v_0$, we have
$$ \vec{u} \le^\star (v_0,\dots,v_{\ell-1}) \wedge (v_0,\dots,v_{\ell-1},v) \le^\star (\vec{u},v).$$
Therefore $\vec{u} = (v_0,\dots,v_{\ell-1})$. In particular, $u = v_{\ell-1}$.

Finally, suppose that $v$ and $w$ are arbitrary non-minimal vertices of $G$, and that $v \prec_B w$. Let $v_\dagger$ and $w_\dagger$ be the second-to-last elements of $\vec{v}$ and $\vec{w}$ respectively. Then $v_\dagger$ and $w_\dagger$ are the $\prec_B$-least neighbors of $v$ and $w$, so it suffices to show that $v_\dagger \preceq_B w_\dagger$.

However, $\vec{v} = (\vec{v}_\dagger,v)$ and $\vec{w} = (\vec{w}_\dagger,w)$. Since $\vec{v} <^\star \vec{w}$, $\vec{v}_\dagger \le^\star \vec{w}_\dagger$, which concludes the proof.
\end{proof}

\end{document}